\documentclass[runningheads]{llncs}
\usepackage[T1]{fontenc}
\usepackage{graphicx}
\usepackage{amsmath,amssymb}
\usepackage{complexity}
\usepackage{hyperref}
\usepackage{enumitem} 
\usepackage{todonotes} 
\usepackage{fullpage}
\newtheorem{observation}{Observation}

\DeclareMathOperator{\preserver}{\mathsf{Preserver}}
\DeclareMathOperator{\destroyer}{\mathsf{Destroyer}}
\DeclareMathOperator{\head}{head}
\DeclareMathOperator{\tail}{tail}

\DeclareMathOperator{\ar}{ar}
\newcommand{\vsat}{\textsc{Planar} $(\le\!3, 3)$-$\mathsf{SAT}$}

\newcommand{\eps}{\varepsilon}

\newcommand{\chim}[1]{\overline{\chi}'_{#1}}

\begin{document}
\title{Maximum edge colouring problem on graphs that exclude a fixed minor\thanks{Supported by project 22-17398S (Flows and cycles in graphs on surfaces) of Czech Science Foundation}}
%
\author{Zden\v {e}k Dvo\v {r}\'ak\inst{1}
\and Abhiruk Lahiri\inst{2}}
\authorrunning{Z. Dvo\v {r}\'ak and A. Lahiri}
%
\institute{Charles University, Prague 11800, Czech Republic\\ 
\email{rakdver@iuuk.mff.cuni.cz} \and
Heinrich Heine University, D\"{u}sseldorf 40225, Germany\\
\email{abhiruk@hhu.de}}
\maketitle              
\begin{abstract}
The maximum edge colouring problem considers the maximum colour assignment to edges of a graph under the condition that every vertex has at most a fixed number of distinct coloured edges incident on it. 
If that fixed number is $q$ we call the colouring a maximum edge $q$-colouring. 
The problem models a non-overlapping frequency channel assignment question on wireless networks. 
The problem has also been studied from a purely combinatorial perspective in the graph theory literature.

We study the question when the input graph is sparse. 
We show the problem remains $\NP$-hard on $1$-apex graphs. 
We also show that there exists $\PTAS$ for the problem on minor-free graphs.
The $\PTAS$ is based on a recently developed Baker game technique for proper minor-closed classes,
thus avoiding the need to use any involved structural results.
This further pushes the Baker game technique beyond the problems expressible in the first-order logic.

\keywords{Polynomial-time approximation scheme  \and Edge colouring \and Minor-free graphs.}
\end{abstract}
\section{Introduction}
For a graph $G = (V, E)$, an \emph{edge} $q$-\emph{colouring} of $G$ is a mapping $f \colon E(G) \to \mathbb{Z}^+$ such that the number of distinct colours incident on any vertex $v \in V(G)$ is bounded by $q$,
and the \emph{spread} of $f$ is the total number of distinct colours it uses.
The \emph{maximum edge} $q$-\emph{chromatic number} $\chim{q}(G)$ of $G$ is the maximum spread of an edge $q$-colouring of $G$.

A more general notion has been studied in the combinatorics and graph theory communities in the context of extremal problems, called \emph{anti-Ramsey number}.
For given graphs $G$ and $H$, the \emph{anti-Ramsey number} $\ar(G, H)$ denotes the maximum number of colours that can be assigned to edges of $G$ so that there does not exist any subgraph isomorphic to $H$ which is \emph{rainbow}, i.e., all the edges of the subgraph receive distinct colours under the colouring. 
The maximum edge $q$-chromatic number of $G$ is clearly equal to $\ar(G, K_{1, q+1})$, where $K_{1, q+1}$ is a star with $q+1$ edges.

The notion of anti-Ramsey number was introduced by Erd{\H{o}}s and Simonovits in 1973 \cite{Erdos}. 
The initial studies focused on determining tight bounds for $\ar(G, H)$. 
A lot of research has been done on the case when $G = K_n$, the complete graph, and $H$ is a specific type of a graph (a path, a complete graph, \ldots)~\cite{MontelN02,Schier04,Simono}. 
For a comprehensive overview of known results in this area, we refer interested readers to~\cite{FujitaMO}.
Bounds on $\ar(K_n, H)$ where $H$ is a star graph are reported in~\cite{Jiang02,ManousSTV96}. 
Gorgol and Lazuka computed the exact value of $\ar(K_n,H)$ when $H$ is  $K_{1,4}$ with an edge added  to it \cite{GorgolL10}. 
For general graph $G$, Montellano-Ballesteros studied $\ar(G, K_{1,q})$ and reported an upper bound~\cite{Montel06}. 

The algorithmic aspects of this problem started gaining attention from researchers around fifteen years ago, due to its application to wireless networks~\cite{Raniwala}.
At that time there was a great interest to increase the capacity of \emph{wireless mesh networks} (which are commonly called \emph{wireless broadband} nowadays).
The solution that became the industry standard is to use multiple channels and transceivers with the ability to simultaneously communicate with many neighbours using multiple radios over the channels~\cite{Raniwala}.
Wireless networks based on the IEEE 802.11a/b/g and 802.16 standards are examples of such systems.
But, there is a physical bottleneck in deploying this solution. 
Enabling every wireless node to have multiple radios can possibly create an interface and thus reduce reliability. 
To circumvent that, there is a limit on the number of channels simultaneously used by any wireless node.
In the IEEE 802.11 b/g standard and IEEE 802.11a standard, the numbers of permittable simultaneous channels are $3$ and $12$, respectively~\cite{WanCWY11}. 

If we model a wireless network as a graph where each wireless node corresponds to a vertex of the graph, then the problem can be formulated as a maximum edge colouring problem. 
The nonoverlapping channels can be associated with distinct colours. 
On each vertex of the graph, the number of distinctly coloured edges allowed to be incident on it captures the limit on the number of channels that can be used simultaneously at each wireless node.
The question of how many channels can be used simultaneously by a given network translates into the number of colours that can be used in a maximum edge colouring.

Devising an efficient algorithm for the maximum edge $q$-colouring problem is not an easy task. 
In \cite{AdamaszekP10}, the problem is reported $\NP$-hard for every $q \geq 2$. 
The authors further showed that the problem is hard to approximate within a factor of $(1+ \frac{1}{q})$ for every $q \geq 2$, assuming the \emph{unique games conjecture}~\cite{AdamP}. 
A simple $2$-approximation algorithm for the maximum edge $2$-colouring problem is reported in \cite{Feng}. 
The same algorithm from \cite{Feng} has an approximation ratio of $5/3$ with the additional assumption that the graph has a perfect matching~\cite{AdamP}. 
It is also known that the approximation ratio can be improved to $8/5$ if the input graph is assumed to be triangle-free~\cite{ChandranLS22}.
An almost tight analysis of the algorithm is known for the maximum edge $q$-colouring problem $(q \geq 3)$ when the input graph satisfies certain degree constraints~\cite{Chandran}. 
The $q=2$ case is also known to be fixed-parameter tractable~\cite{GoyalKM13}.

In spite of several negative theoretical results, the wireless network question continued drawing the attention of researchers due to its relevance in applications.
There are several studies focusing on improving approximation under further assumptions on constraints that are meaningful in practical applications~\cite{KodialamN05}, \cite{SenMGB07}, \cite{WanCWY11}, \cite{WanAJWX15}. 
This motivates us to study the more general question on a graph class that captures the essence of wireless mesh networks.
Typically, disk graphs and unit disk graphs are well-accepted abstract models for wireless networks.
But they can capture more complex networks than what a real-life network looks like~\cite{SenMGB07}.
By definition, both unit disk graphs and disk graphs can have arbitrary size cliques.
In a practical arrangement of a wireless mesh network, it is quite unlikely to place too many wireless routers in a small area.
In other words, a real-life wireless mesh network can be expected to be fairly sparse and avoid large cliques.
In this paper, we focus on a popular special case of sparse networks, those avoiding a fixed graph as a minor.
In particular, this includes the graphs that can be made planar by deletion of a bounded number $k$ of vertices (the \emph{$k$-apex graphs}).

From a purely theoretical perspective, the graphs avoiding a fixed minor are interesting on their own merit.
Famously, they admit the structural decomposition devised by Robertson and Seymour~\cite{RobertsonS03a},
but also have many interesting properties that can be shown directly, such as the existence of
sublinear separators~\cite{alon1990separator} and admitting layered decomposition into pieces of bounded weak diameter~\cite{klein1995}.
They have been also intensively studied from the algorithmic perspective, including the $\PTAS$ design.
Several techniques for this purpose have been developed over the last few decades. 
The \emph{bidimensionality technique} bounds the treewidth of the graph in terms of the size of the optimal solution and uses the balanced separators to obtain the approximation factor~\cite{DemaineH05,FominLRS11}.
A completely different approach based on local search is known for unweighted problems~\cite{CabelloG15,Har-PeledQ17}.
Dvo\v{r}{\'{a}}k used thin systems of overlays~\cite{Dvorak18} and a generalization of Baker's layering approach~\cite{baker1994approximation,Dvorak20}
to obtain $\PTAS$es for a wide class of optimization problems expressible in the first-order logic and its variations.

\subsection{Our results}
Our contribution is twofold. 
First, we show that the maximum edge $q$-colouring problem is $\NP$-hard on $1$-apex graphs. 
Our approach is similar in spirit to the approximation hardness reduction for the problem on general graphs~\cite{AdamaszekP10}.  

Secondly, we show that there exists a $\PTAS$ for the maximum edge $q$-colouring problem for graphs avoiding a fixed minor.
The result uses the \emph{Baker game} approach devised in~\cite{Dvorak20}, avoiding the use of involved structural results.
The technique was developed to strengthen and simplify the results of~\cite{DawarGKS06} giving $\PTAS$es for monotone optimization problems expressible in the first-order logic.  
Our work demonstrates the wider applicability of this technique to problems not falling into this framework.

\section{Preliminaries}

A graph $H$ is a \emph{minor} of a graph $G$ if a graph isomorphic to $H$ can be obtained from a subgraph of $G$ by a series of edge contractions.  
We say that $G$ is \emph{$H$-minor-free} if $G$ does not contain $H$ as a minor.
A graph is called \emph{planar} if it can be drawn in the plane without crossings.
A graph $G$ is a \emph{$k$-apex graph} if there exists a set $A\subseteq V(G)$ of size at most $k$ such that $G-A$ is planar.
The $k$-apex graphs are one of the standard examples of graphs avoiding a fixed minor; indeed, they are $K_{k+5}$-minor-free.

Given a function $f$ assigning colours to edges of a graph $G$ and a vertex $v\in V(G)$, we write $f(v)$ to denote the set $\{f(e) :\text{$e$ is adjacent to $v$}\}$, and $f(G)=\{f(e):e\in E(G)\}$.  
Recall that $f$ is an edge $q$-colouring of $G$ if and only if $|f(v)|\le q$ for every $v\in V(G)$, and the maximum edge $q$-chromatic number of $G$ is
$$\chim{q}(G)=\max\{|f(G)|:\text{$f$ is an edge $q$-colouring of $G$}\}.$$

A \emph{matching} in a graph $G$ is a set of edges of $G$ where no two are incident with the same vertex.
A matching $M$ is \emph{maximal} if it is not a proper subset of any other matching.
Note that a maximal matching is not necessarily the largest possible.
Let $|G|$ denote $|V(G)|+|E(G)|$.
For all other definitions related to graphs not defined in this article, we
refer readers to any standard graph theory textbook, such as~\cite{Diestel}.

\section{\PTAS~for minor-free graphs}

Roughly speaking, we employ a divide-and-conquer approach to approximate $\chim{q}(G)$, splitting $G$ into vertex disjoint parts $G_1$, \ldots, $G_m$ in a suitable way, solving the problem for each part recursively, and combining the solutions. 
An issue that we need to overcome is that it may be impossible to compose the edge $q$-colourings,
e.g., if an edge $(v_1,v_2)$ joins distinct parts and disjoint sets of $q$ colours are used on the neighbourhoods of $v_1$ and $v_2$ already.  
To overcome this issue, we reserve the colour $0$ to be used to join the ``boundary'' vertices.  
This motivates the following definition.

For a set $S$ of vertices of a graph $G$, an edge $q$-colouring $f$ is \emph{$S$-composable} if $|f(v)\setminus\{0\}|\le q-1$ for every $v\in S$.  Let $\chim{q}(G,S)$ denote the maximum number of non-zero colours that can be used by an $S$-composable edge $q$-colouring of $G$.
Let us remark that $G$ has an $S$-composable edge $q$-colouring using any non-negative number $k'\le \chim{q}(G,S)$ of non-zero colours, as all edges of any colour $c\neq 0$ can be recoloured to $0$.

\begin{observation}
\label{obs:max-q-col-with-saturated-col}
For any graph $G$, we have $\chim{q}(G)=\chim{q}(G,\emptyset)$, and $\chim{q}(G,$ $S)\le \chim{q}(G)$ for any $S\subseteq V(G)$.
\end{observation}

We need the following approximation for $\chim{q}(G,S)$ in terms of the size of a maximal matching, analogous to one for edge $2$-colouring given in~\cite{Feng}.  
Let us remark that the $S$-composable edge $q$-colouring problem is easy to solve for $q=1$, since we have to use colour $0$ on all edges of each component intersecting $S$ and we can use a distinct colour for all edges of any other component.  
Consequently, in all further claims, we assume $q\ge 2$.

\begin{observation}
\label{obs:sat-col-matching}
For any graph $G$, any $S\subseteq V(G)$, any maximal matching $M$ in $G$, and any $q\ge 2$,
$$|M| \leq \chim{q}(G,S) \leq \chim{q}(G)\le 2q|M|.$$
\end{observation}
\begin{proof}
We can assign to each edge of $M$ a distinct positive colour and to all other edges (if any) the colour $0$, obtaining an $S$-composable edge $2$-colouring using $|M|$ non-zero colours.  
On the other hand, consider the set $X$ of vertices incident with the edges of $M$.  
By the maximality of $M$, the set $X$ is a vertex cover of $G$, i.e., each edge of $G$ is incident with a vertex of $X$, and thus at most $q|X|=2q|M|$ colours can be used by any edge $q$-colouring of $G$.\qed
\end{proof}

In particular, as we show next, the lower bound implies that the $S$-composable edge $q$-colouring problem is fixed-parameter tractable when parameterized by the value of the solution
(a similar observation on the maximum edge $2$-colouring is reported in~\cite{GoyalKM13}).

\begin{observation}\label{obs:fpt}
There exists an algorithm that, given a graph $G$, a set $S\subseteq V(G)$, and integers $q\ge 2$ and $s$, in time $O_{q,s}(|G|)$ returns an $S$-composable edge $q$-colouring of $G$ using at least $\min(\chim{q}(G,S),s)$
colours.
\end{observation}

\begin{proof}
We can in linear time find a maximal matching $M$ in $G$.  If $|M|\ge s$, we
return the colouring that gives each edge of $M$ a distinct non-zero colour and all other edges colour $0$.  
Otherwise, the set $X$ of vertices incident with $M$ is a vertex cover of $G$ of size at most $2s-2$, and thus $G$ has treewidth at most $2s-2$.  
Note also that for any $s'$, there exists a formula $\varphi_{s',q}$ in monadic second-order logic such that $G,S,E_0,\ldots, E_{s'}\models \varphi_{s',q}$
if and only if $E_0$, \ldots, $E_{s'}$ is a partition of the edges of $G$ with all parts except possibly for $E_0$ non-empty such that the function $f$ defined by letting $f(e)=i$ for each $i\in\{0,\ldots,s'\}$ and $e\in E_i$ is an $S$-composable edge $q$-colouring of $G$.  
Therefore, we can find an $S$-composable edge $q$-colouring of $G$ with the maximum number $s'\le s$ of non-zero colours using Courcelle's theorem~\cite{courcelle} in time $O_{q,s}(|G|)$.\qed
\end{proof}

A \emph{layering} of a graph $G$ is a function $\lambda \colon V(G) \to \mathbb{Z}^+$ such that $|\lambda(u) - \lambda(v)| \leq 1$ for every edge $(u, v) \in E(G)$. 
In other words, the graph is partitioned into layers $\lambda^{-1}(i)$ for $i \in \mathbb{Z}^+$ such that edges of $G$ only appear within the layers and between the consecutive layers.
Baker~\cite{baker1994approximation} gave a number of $\PTAS$es for planar graphs based on the fact that in a layering of a connected planar graph according to the distance from a fixed vertex, the union of a constant number of consecutive layers induces a subgraph of bounded treewidth.  
This is not the case for graphs avoiding a fixed minor in general, however, a weaker statement expressed in terms of Baker game holds.  
We are going to describe that result in more detail in the following subsection.
Here, let us state the key observation that makes layering useful for approximating the edge $q$-chromatic number.

For integers $r\ge 2$ and $m$ such that $0\le m\le r-1$, the \emph{$(\lambda,r,m)$-stratification} of a graph $G$
is the pair $(G',S')$ such that
\begin{itemize}
\item $G'$ is obtained from $G$ by deleting all edges $uv$ such that $\lambda(u)\equiv m\pmod r$ and $\lambda(v)\equiv m+1\pmod r$, and
\item $S'$ is the set of vertices of $G$ incident with the edges of $E(G)\setminus E(G')$.
\end{itemize}

\begin{lemma}
\label{lem:approx}
Let $G$ be a graph, $S$ a subset of its vertices, and $q,r\ge 2$ integers.  
Let $\lambda$ be a layering of $G$.
For $m\in\{0,\ldots,r-1\}$, let $(G_m,S_m)$ be the $(\lambda,r,m)$-stratification of $G$.
\begin{itemize}
\item $\chim{q}(G_m,S\cup S_m)\le \chim{q}(G,S)$ for every $m\in\{0,\ldots,r-1\}$.
\item There exists $m\in\{0,\ldots,r-1\}$ such that $\chim{q}(G_m,S\cup S_m)\ge \Bigl(1-\frac{6q}{r}\Bigr)\chim{q}(G,S)$.
\end{itemize}
\end{lemma}
\begin{proof}
Given an $(S\cup S_m)$-composable edge $q$-colouring of $G_m$, we can assign the colour $0$ to all edges of $E(G)\setminus E(G_m)$ and obtain an $S$-composable edge $q$-colouring of $G$ using the same number of non-zero colours, which implies that $\chim{q}(G_m, S\cup S_m)\leq \chim{q}(G,S)$.

Conversely, consider an $S$-composable edge $q$-colouring $f$ of $G$ using $k=\chim{q}(G,S)$ non-zero colours.
For $m\in\{0,\ldots,r-1\}$, let $B_m$ be the bipartite graph with vertex set $S_m$ and edge set $E(G)\setminus E(G_m)$ and let $M_m$ be a maximal matching in $B_m$.  
Let $\mathcal{P}$ be a partition of the set $\{0,\ldots, r-1\}$ into at most three disjoint parts such that none of the parts contains two integers that are consecutive modulo $r$.
For each $P\in \mathcal{P}$, let $M_P=\bigcup_{m\in P} M_m$, and observe that $M_P$ is a matching in $G$.
By Observation~\ref{obs:sat-col-matching}, it follows that $k\ge |M_P|$, and thus
$$3k\ge |P|k\ge \sum_{P\in\mathcal{P}} |M_P|=\sum_{m=0}^{r-1}|M_m|.$$
Hence, we can fix $m\in\{0,\ldots,r-1\}$ such that $|M_m|\le \tfrac{3}{r}k$.
By Observation~\ref{obs:sat-col-matching}, any edge $q$-colouring of $B_m$, and in particular the restriction of $f$ to the edges of $B_m$, uses at most $2q|M_m|\le \tfrac{6q}{r}k$ distinct colours.

Let $f'$ be the edge $q$-colouring of $G$ obtained from $f$ by recolouring all edges whose colour appears on the edges of $B_m$ to colour $0$.  
Clearly $f'$ uses at least $\bigl(1-\tfrac{6q}{r}\bigr)k$ non-zero colours.
Moreover, each vertex $v\in S_m$ is now incident with an edge of colour $0$, and thus $|f'(v)\setminus\{0\}|\le q-1$.  
Therefore, the restriction of $f'$ to $E(G_m)$ is an $(S\cup S_m)$-composable edge $q$-colouring, implying that
$$\chim{q}(G_m,S\cup S_m)\ge \Bigl(1-\frac{6q}{r}\Bigr)k=\Bigl(1-\frac{6q}{r}\Bigr)\chim{q}(G,S).$$ 
\qed
\end{proof}

Hence, if $r\gg q$, then a good approximation of $\chim{q}(G_m,S\cup S_m)$ for all $m\in\{0,\ldots,r-1\}$ gives a good approximation for $\chim{q}(G,S)$.  
We will also need a similar observation for vertex deletion; here we only get an additive approximation in general, but as long as the edge $q$-chromatic number is large enough, this suffices (and if it is not, we can determine it exactly using Observation~\ref{obs:fpt}).

\begin{lemma}\label{lem:delv}
Let $G$ be a graph, $S$ a set of its vertices, and $v$ a vertex of $G$.  Let $S'=(S\setminus\{v\})\cup N(v)$.
For any integer $q\ge 2$, we have
$$\chim{q}(G,S)\ge \chim{q}(G-v,S')\ge \chim{q}(G,S)-q,$$ and in particular if $\varepsilon>0$ and $\chim{q}(G,S)\ge q/\varepsilon$, then 
$$\chim{q}(G-v,S')\ge (1-\varepsilon)\chim{q}(G,S).$$
\end{lemma}
\begin{proof}
Any $S'$-composable edge $q$-colouring of $G-v$ extends to an $S$-composable edge $q$-colouring of $G$ by giving all edges incident on $v$ colour $0$, implying that $\chim{q}(G,S)\ge \chim{q}(G-v,S')$.
Conversely, any $S$-composable edge $q$-colouring of $G$ can be turned into an $S'$-composable edge $q$-colouring of $G-v$ by recolouring all edges whose colour appears on the neighbourhood of $v$ to $0$ and restricting it to the edges of $G-v$.
This loses at most $q$ non-zero colours (those appearing on the neighborhood of $v$), and thus $\chim{q}(G-v,S')\ge \chim{q}(G,S)-q$.\qed
\end{proof}

\subsection{Baker game} 

For an infinite sequence $\mathtt{r} = r_1, r_2, \dots$ and an integer $s \geq 0$, let $\tail(\mathtt{r})$ denote the sequence $r_2, r_3, \dots$ and let $\head(\mathtt{r}) = r_1$. 
\emph{Baker game} is played by two players $\destroyer$ and $\preserver$ on a pair $(G,\mathtt{r})$, where $G$ is a graph and $\mathtt{r}$ is a sequence of positive integers.  
The game stops when $V(G) = \emptyset$, and $\destroyer$'s objective is to minimise the number of rounds required to make the graph empty. 
In each round of the game, either
\begin{itemize}
\item $\destroyer$ chooses a vertex $v\in V(G)$, $\preserver$ does nothing and the game moves to the state $(G \setminus \{v\}, \tail(\mathtt{r}))$, or
\item $\destroyer$ selects a layering $\lambda$ of $G$, $\preserver$ selects an interval $I$ of $\head(\mathtt{r})$ consecutive integers and the game moves to the state $(G[\lambda^{-1}(I)],$ $\tail(\texttt{r}))$.  In other words, $\preserver$ selects $\head(\mathtt{r})$ consecutive layers and the rest of the graph is deleted.
\end{itemize}

$\destroyer$ \emph{wins} in $k$ rounds on the state $(G, \mathtt{r})$ if regardless of $\preserver$'s strategy, the game stops after at most $k$ rounds. 
As we mentioned earlier $\destroyer$'s objective is to minimise the number of rounds of this game and it is known that they will succeed if the game is played on a graph that forbids a fixed minor
(the upper bound on the number of rounds depends only on the sequence $\mathtt{r}$ and the forbidden minor, not on $G$).
\begin{theorem}[Dvo\v{r}\'ak~\cite{Dvorak20}]
\label{thm:baker}
For every graph $F$ and every sequence $\mathtt{r} = r_1, r_2, \dots$ of positive integers, there exists a positive integer $k$ such that for every graph $G$ avoiding $F$ as a minor, $\destroyer$ wins Baker game from the state $(G,\mathtt{r})$ in at most $k$ rounds.
Moreover, letting $n=|V(G)|$, there exists an algorithm that preprocesses $G$ in time $O_F(n^2)$ and then in each round determines a move for $\destroyer$ (leading to winning in at most $k$ rounds in total) in time $O_{F,\mathtt{r}}(n)$.
\end{theorem}

Let us now give the algorithm for approximating the edge $q$-chromatic number on graphs for which we can quickly win Baker game.
\begin{lemma}\label{lem:main}
There exists an algorithm that, given
\begin{itemize}
\item a graph $G$, a set $S\subseteq V(G)$, an integer $q\ge 2$, and
\item a sequence $\mathtt{r} = r_1, r_2, \dots$ of positive integers such that $\destroyer$ wins Baker game from the state
$(G,\mathtt{r})$ in at most $k$ rounds, and in each state that arises in the game is able to determine the move that achieves this in time $T$,
\end{itemize}
returns an $S$-composable edge $q$-colouring of $G$
using at least $\Bigl(\prod_{i=1}^k \bigl(1-\tfrac{6q}{r_i}\bigr)\Bigr)\cdot \chim{q}(G,S)$ non-zero colours, in time $O_{\mathtt{r},k,q}(|G|T)$.
\end{lemma}
\begin{proof}
First, we run the algorithm from Observation~\ref{obs:fpt} with $s=\lceil r_1/3\rceil$.
If the obtained colouring uses less than $s$ non-zero colours, it is optimal and we return it.
Otherwise, we know that $\chim{q}(G,S)\ge s$.  In particular, $E(G)\neq\emptyset$, and thus $\destroyer$ have not won the game yet.

Let $R=\Bigl(\prod_{i=2}^k \bigl(1-\tfrac{6q}{r_i}\bigr)\Bigr)$.
Let us now consider two cases depending on $\destroyer$'s move from the state $(G,\mathtt{r})$.
\begin{itemize}
\item Suppose that $\destroyer$ decides to delete a vertex $v\in V(G)$.
We apply the algorithm recursively for the graph $G-v$, set $S'=(S\setminus\{v\})\cup N(v)$, and the sequence $\tail(\texttt{r})$, obtaining an $S'$-composable edge $q$-colouring $f$ of $G-v$ using at least $R\cdot\chim{q}(G-v,S')$ non-zero colours.
By Lemma~\ref{lem:delv} with $\varepsilon=\tfrac{q}{s}$, we conclude that $f$ uses at least 
$$R\cdot\chim{q}(G-v,S')\ge R(1-\varepsilon)\chim{q}(G,S)\ge R\Bigl(1-\frac{6q}{r_1}\Bigr)\chim{q}(G,S)$$
non-zero colours.  
We turn $f$ into an $S$-composable edge $q$-colouring of $G$ by giving all edges incident on $v$ colour $0$ and return it.
\item Suppose that $\destroyer$ chooses a layering $\lambda$.  We now recurse into several subgraphs, each
corresponding to a valid move of $\preserver$.  For each $m\in\{0,\ldots, r_1-1\}$, let
$(G_m,S_m)$ be the $(\lambda,r_1,m)$-stratification of $G_m$.  
Note that $G_m$ is divided into parts $G_{m,1}$, \ldots, $G_{m,t_m}$, each contained in the union of $r_1$ consecutive layers of $\lambda$.  
For each $m\in\{0,\ldots,r_1-1\}$ and each $i\in\{1,\ldots,t_m\}$, we apply the algorithm recursively for the graph $G_{m,i}$, set $S_{m,i}=(S_m\cup S)\cap V(G_{m,i})$, and the sequence $\tail(\texttt{r})$, obtaining an $S_{m,i}$-composable edge $q$-colouring $f_{m,i}$ of $G_{m,i}$ using at least $R\cdot\chim{q}(G_{m,i},S_{m,i})$ non-zero colours.  
Let $f_m$ be the union of the colourings $f_{m,i}$ for $i\in \{1,\ldots,t_m\}$ and observe that $f_m$ is an $(S\cup S_m)$-composable edge $q$-colouring of $G_m$ using at least $R\cdot\chim{q}(G_m,S\cup S_m)$ non-zero colours.  
We choose $m\in\{0,\ldots, r_1-1\}$ such that $f_m$ uses the largest number of non-zero colours, extend it to an $S$-composable edge $q$-colouring of $G$ by giving all edges of $E(G)\setminus E(G_m)$ colour $0$, and return it.
By Lemma~\ref{lem:approx}, the colouring uses at least
$$R\cdot\chim{q}(G_m,S\cup S_m)\ge R\Bigl(1-\frac{6q}{r_1}\Bigr)\chim{q}(G,S)$$ non-zero colours, as required.
\end{itemize}
For the time complexity, note that each vertex and edge of $G$ belongs to at most $\prod_{i=1}^d r_i$ subgraphs processed at depth $d$ of the recursion, and since the depth of the recursion is bounded by $k$, the sum of the sizes of the processed subgraphs is $O_{\mathtt{r},k,q}(|G|)$.  
Excluding the recursion and time needed to select $\destroyer$'s moves, the actions described above can be performed in linear time.  Consequently, the total time complexity
is $O_{\mathtt{r},k,q}(|G|T)$.\qed
\end{proof}

Our main result is then just a simple combination of this lemma with Theorem~\ref{thm:baker}.

\begin{theorem}
\label{thm:main}
There exists an algorithm that given an $F$-minor-free graph $G$ and integers $q,p\ge 2$, returns in time $O_{F,p,q}(|G|^2)$ an edge $q$-colouring of $G$ using at least $(1-1/p)\chim{q}(G)$ colours.
\end{theorem}
\begin{proof}
Let $\mathtt{r}$ be the infinite sequence such that $r_i = 10pqi^2$ for each positive integer $i$, and let $k$ be the number of rounds in which $\destroyer$ wins Baker game from the state $(G',\mathtt{r})$ for any $F$-minor-free graph $G'$, using the strategy given by Theorem~\ref{thm:baker}.
Note that
\begin{align*}
R&=\prod_{i=1}^k \bigl(1-\tfrac{6q}{r_i}\bigr)\ge 1-\sum_{i=1}^\infty \frac{6q}{r_i}\\
&=1-\frac{3}{5p}\sum_{i=1}^\infty \frac{1}{i^2}=1-\frac{3}{5p}\cdot \frac{\pi^2}{6}\ge 1-\frac{1}{p}.
\end{align*}
Let $n=|G|$.  
After running the preprocessing algorithm from Theorem~\ref{thm:baker}, we apply the algorithm from Lemma~\ref{lem:main} with $S=\emptyset$ and $T=O_{F,\mathtt{r}}(n)=O_{F,p,q}(n)$, obtaining an edge $q$-colouring of $G$ using at least $R\cdot\chim{q}(G,\emptyset)=R\cdot\chim{q}(G)\ge (1-1/p)\chim{q}(G)$ colours, in time $O_{F,p,q}(n^2)$.\qed
\end{proof}

\section{Hardness on \texorpdfstring{$1$}{1}-apex graphs}
In this section, we study the complexity of the maximum edge $2$-colouring problem on $1$-apex graphs. 
We present a reduction from \vsat{} which is known to be $\NP$-hard~\cite{Tovey84}. 

The \emph{incidence graph} $G(\varphi)$ of a Boolean formula $\varphi$ in conjunctive normal form is the bipartite graph whose vertices are the variables appearing in $\varphi$ and the clauses of $\varphi$, and each variable is adjacent exactly to the clauses in which it appears.
A Boolean formula $\varphi$ in conjunctive normal form is called \vsat{} if 
\begin{itemize}
\item each clause of $\varphi$ contains at most three distinct literals,
\item each variable of $\varphi$ appears in exactly three clauses,
\item the incidence graph $G(\varphi)$ is planar.
\end{itemize}
In \vsat{} problem, we ask whether such a formula $\varphi$ has a satisfying assignment.

We follow the strategy used in~\cite{AdamaszekP10}, using an intermediate
\emph{maximum edge} $1,2$-\emph{colouring} problem.
The instance of this problem consists of a graph $G$, a function $g \colon V(G) \rightarrow \{1, 2\}$, and a number $t$.
An \emph{edge $g$-colouring} of $G$ is an edge colouring $f$ such that $|f(v)|\le g(v)$ for each $v\in V(G)$.
The objective is to decide whether there exists an edge $g$-colouring of $G$ using at least $t$ distinct colours.
We show the maximum edge $\{1,2\}$-colouring problem is $\NP$-hard on $1$-apex graphs by establishing a reduction from \vsat{} problem. 
We then use this result to show that the maximum edge $q$-colouring problem on planar graphs is $\NP$-hard when $q \geq 2$.
Let us start by establishing the intermediate result.

\begin{lemma}
\label{lem:hardness}
The maximum edge $\{1, 2\}$-colouring problem is $\NP$-hard even when restricted on the class of $1$-apex graphs. 
\end{lemma}
\begin{proof}
Consider a given \vsat{} formula $\varphi$ with $m$ clauses and $n$ variables and a plane drawing of its incidence graph $G(\varphi)$.
Let the clauses of $\varphi$ be $c_1$, \ldots, $c_m$ and the variables $x_1$, \ldots, $x_n$; we use the same symbols for the corresponding vertices of $G(\varphi)$.

Let $H$ be a graph obtained from $G(\varphi)$ as follows.
For all $j \in \{1,2, \dots, n\}$, if the clauses in which $x_j$ appears are $c_{\ell_{j,1}}$, $c_{\ell_{j,2}}$, and $c_{\ell_{j,3}}$, split $x_j$ to three vertices $x_{j,1}$, $x_{j,2}$, and $x_{j,3}$, where $x_{j,a}$ is adjacent to $c_{\ell_{j,a}}$ for $a\in \{1,2,3\}$.  
For $1\le a<b\le 3$, add a vertex $n_{j,a,b}$ and if $x_j$ appears positively in $c_{\ell_{j,a}}$ and negatively in $c_{\ell_{j,b}}$ or vice versa, make it adjacent to $x_{j,a}$ and $x_{j,b}$ (otherwise leave it as an isolated vertex).
Finally, we add a new vertex $u$ adjacent to $c_i$ for $i\in\{1,\ldots,m\}$ and to $n_{j,a,b}$ for $j\in\{1,\ldots,n\}$ and $1\le a<b\le 3$.
Clearly, $H$ is a $1$-apex graph, since $H-u$ is planar.

Let us define the function $g \colon V(H) \rightarrow \{1, 2\}$ as follows:
\begin{itemize}
\item $g(u) = 1$,
\item $g(c_i) = 2$ for all $i \in \{1, 2, \dots , m\}$,
\item $g(x_{j,a}) = 1$ for all $j \in \{1, 2, \dots , n\}$ and $a\in\{1,2,3\}$, and
\item $g(n_{j,a,b}) = 2$ for all $j \in \{1, 2, \dots , n\}$ and $1\le a < b\le 3$.
\end{itemize}

First, we show if there exists a satisfying assignment for the formula $\varphi$, then $H$ has an edge $g$-colouring using $n+1$ colours.
For $i\in\{1,\ldots, n\}$, choose a vertex $x_{j,a}$ adjacent to $c_i$ such that the (positive or negative) literal of $c_i$ containing the variable $x_j$ is true in the assignment, and give colour $i$ to the edge $(c_i,x_{j,a})$ and all other edges incident on $x_{j,a}$ (if any).  
All other edges receive colour $0$.

Clearly, $u$ is only incident with edges of colour $0$, for each $j\in\{1,\ldots,n\}$ and $a\in\{1,\ldots,3\}$ all edges incident on $x_{j,a}$ have the same colour, and for each $i\in\{1,\ldots,m\}$, the edges incident on $c_i$ have colours $0$ and $i$.  
Finally, consider a vertex $n_{j,a,b}$ for some $j\in\{1,\ldots,n\}$ and $1\le a<b\le 3$ adjacent to $x_{j,a}$ and $x_{j,b}$.  
By the construction of $H$, the variable $x_j$ appears positively in $c_{\ell_{j,a}}$ and negatively in $c_{\ell_{j,b}}$ or vice versa, and thus at most one of the corresponding literals is true in the assignment.  
Hence, $n_{j,a,b}$ is incident with edges of colour $0$ and of at most one of the colours $\ell_{j,a}$ and $\ell_{j,b}$.

Conversely, suppose that there exists an edge $g$-colouring $f$ of $H$
using at least $m+1$ distinct colours, and let us argue that there exists a satisfying assignment for $\varphi$.
Since $g(u)=1$, we can without loss of generality assume that each
edge incident with $u$ has colour $0$.  
If a colour $c\neq 0$ is used to colour the edge $(n_{j,a,b}, x_{j,k})$ for some $j\in\{1,\ldots,n\}$, $1\le a <b\le 3$, and $k\in\{a,b\}$, then since $g(x_{j,k}) = 1$, this colour is also used on the edge $(x_{j,k}, c_{\ell_{j,k}})$. 
Hence, every non-zero colour appears on an edge incident with a clause.  
Since each clause is also joined to $u$ by an edge of colour $0$, it can be only incident with edges of one other colour.  
Since $f$ uses at least $m+1$ colours, we can without loss of generality assume that for $i\in\{1,\ldots,m\}$, there exists an edge $(c_i, x_{j,a})$ for some $j\in\{1,\ldots,n\}$ and $a\in\{1,\ldots,3\}$ of colour $i$.
Assign to $x_j$ the truth value that makes the literal of $c_i$ in which it appears true.

We only need to argue that this rule does not cause us to assign to $x_j$ both values true and false.
If that were the case, then there would exist $1\le a<b\le 3$ such that the variable $x_j$ appears positively in clause $c_{\ell_{j,a}}$ and negatively in clause $c_{\ell_{j,b}}$ or vice versa, the edge corresponding to the variable $x_{j,a}$ has colour $\ell_{j,a}$ and the edge corresponding to the variable $x_{j,b}$ has colour $\ell_{j,b}$.  
However, since $g(x_{j,a})=g(x_{j,b})=1$, this would imply that $n_{j,a,b}$ is incident with the edge $(n_{j,a,b}, x_{j,a})$ of colour $\ell_{j,a}$, the edge $(n_{j,a,b},x_{j,b})$ of colour $\ell_{j,b}$, and the edge $(n_{j,a,b}, u)$ of colour $0$, which is a contradiction.

Therefore, we described how to transform in polynomial time a \vsat{} instance $\varphi$ to an equivalent instance $H$, $g$, $t=m+1$ of the maximum edge $\{1,2\}$-colouring problem.\qed
\end{proof}

Now we are ready to prove the main theorem of this section.
The proof strategy is similar to the $\APX$-hardness proof in~\cite{AdamaszekP10}.
We include the details for completeness. 
\begin{theorem}
\label{thm:hardness}
For an arbitrary integer $q \geq 2$ the maximum edge $q$-colouring problem is $\NP$-hard even when the input instance is restricted to $1$-apex graphs. 
\end{theorem}

\begin{proof}
We construct a reduction from the maximum edge $\{1,2\}$-colouring problem on $1$-apex graphs.
Let $G$, $g$, $t$ be an instance of this problem, and let $n=|V(G)|$ and $r=|\{v \in V(G) \colon g(v) = 1\}|$.
We create a graph $G'$ from $G$ by adding for each vertex $v \in V(G)$ exactly $q - g(v)$ pendant vertices adjacent to $v$.
Clearly, $G'$ is an $1$-apex graph.  We show that $G$ has an edge $g$-colouring using at least $t$ distinct colours if and only if $G'$ has an edge $q$-colouring using at least $t + r + (q - 2)n$ colours.

In one direction, given an edge $g$-colouring of $G$ using at least $t$ colours,  we colour each of the added pendant edges using a new colour, obtaining an edge $q$-colouring of $G$ using at least $t + r + (q - 2)n$ colours.

Conversely, let $f$ be an edge $q$-colouring of $G'$ using at least $t + r + (q - 2)n$ colours.
Process the vertices $v\in V(G)$ one by one, performing for each of them the following operation:
For each added pendant vertex $u$ adjacent to $v$ in order, let $c'$ be the colour of the edge $(u,v)$, delete $u$, and if $v$ is incident with an edge $e$ of colour $c\neq c'$, then recolour all remaining edges of colour $c'$ to $c$.  
Note that the number of eliminated colours is bounded by the number $r + (q - 2)n$ of pendant vertices, and thus the resulting colouring still uses at least $t$ colours.  
Moreover, at each vertex $v\in V(G)$, we either end up with all edges incident on $v$ having the same colour or we eliminated one colour from the neighbourhood of $v$ for each adjacent pendant vertex; in the latter case, since $|f(v)|\le q$ and $v$ is adjacent to $q-g(v)$ pendant vertices, at most $g(v)$ colours remain on the edges incident on $v$.
Hence, we indeed obtain an edge $g$-colouring of $G$ using at least $t$ colours. \qed
\end{proof}

\section{Future directions}
We conclude with some possible directions for future research.
The maximum edge $2$-colouring problem on $1$-apex graphs is $\NP$-hard. 
But the complexity of the problem is unknown when the input is restricted to planar graphs. 
We consider this an interesting question left unanswered.
The best-known approximation ratio is known to be $2$, without any restriction on the input instances. 
Whereas, a lower bound of $(\frac{1+q}{q})$, for $q \geq 2$ is known assuming unique games conjecture. 
There are not many new results reported in the last decade that bridge this gap.
We think, even a $(2 - \eps)$ algorithm, for any $\varepsilon > 0$, will be a huge progress towards that direction. 
The Baker game technique can yield $\PTAS$es for monotone optimization problems beyond problems expressible in the first-order logic.
Clearly, the technique can't be extended to the entire class of problems expressible in the monadic second-order logic.
It will be interesting to characterise the problems expressible in the monadic second-order logic where the Baker game yield $\PTAS$es.

\section{Acknowledgement}
The second author likes to thank Benjamin Moore, Jatin Batra, Sandip Banerjee and Siddharth Gupta for helpful discussions on this project. 
He also likes to thank the organisers of Homonolo for providing a nice and stimulating research environment. 
%
%
\bibliographystyle{splncs04}

\end{document}